\documentclass[aps,prl,twocolumn,superscriptaddress,10pt]{revtex4-2}

\usepackage{amsthm}
\usepackage{amssymb}
\usepackage{amsmath}
\usepackage{amsfonts}
\usepackage{physics}
\usepackage{graphicx}
\usepackage{dcolumn}
\usepackage{bm}
\usepackage[utf8]{inputenc}
\usepackage[table]{xcolor}
\usepackage[T1]{fontenc}
\usepackage[colorlinks=true,citecolor=teal,linkcolor=teal]{hyperref}
\hypersetup{allcolors={teal}}
\usepackage{tikz}
\usepackage{lipsum}
\usepackage{epsfig}
\usepackage{epstopdf}
\usepackage{mathtools}
\usepackage{braket}
\usepackage{hyperref}
\usepackage[capitalise]{cleveref}
\usepackage{amsthm}
\usepackage{enumitem}
\usepackage{dsfont}
\usepackage{bbold}
\usepackage{eurosym}
\usepackage{diagbox}
\usepackage{comment}
\usepackage{algorithm}
\usepackage{algpseudocode}
\usepackage{enumitem} 
    \newlist{sublist}{enumerate}{1}
    \setlist[sublist,1]{label=(\alph*)}

\crefname{section}{section}{sections}
\newcommand{\ba}{\begin{aligned}}
\newcommand{\ea}{\end{aligned}}

\newcommand{\bc}{\begin{center}}
\newcommand{\ec}{\end{center}}
\newcommand{\beq}{\begin{equation}}
\newcommand{\eeq}{\end{equation}}
\newcommand{\beqq}{\begin{equation*}}
\newcommand{\eeqq}{\end{equation*}}
\newcommand{\beqa}{\begin{align}}
\newcommand{\eeqa}{\end{align}}
\newcommand{\barr}{\begin{array}}
\newcommand{\earr}{\end{array}}
\newcommand{\bi}{\begin{itemize}}
\newcommand{\ei}{\end{itemize}}

\newcommand{\bC}{\mathbb C}
\theoremstyle{plain}
\newcounter{tho}
\newtheorem{theo}[tho]{Theorem}
\newtheorem{lemma}[tho]{Lemma}
\newtheorem{cor}[tho]{Corollary}

\begin{document}

 \title{On the complex zeros of the wavefunction}

\author{Sacha Cerf}
\affiliation{QAT team, DIENS, \'Ecole Normale Sup\'erieure, PSL University, CNRS, INRIA, 45 rue d'Ulm, Paris 75005, France}

\author{Clara Wassner}
\affiliation{Dahlem Center for Complex Quantum Systems, Freie Universit\"at, 14195 Berlin, Germany}

\author{Jack Davis}
\affiliation{QAT team, DIENS, \'Ecole Normale Sup\'erieure, PSL University, CNRS, INRIA, 45 rue d'Ulm, Paris 75005, France}

\author{Francesco Arzani}
\affiliation{QAT team, DIENS, \'Ecole Normale Sup\'erieure, PSL University, CNRS, INRIA, 45 rue d'Ulm, Paris 75005, France}

\author{Ulysse Chabaud}
\affiliation{QAT team, DIENS, \'Ecole Normale Sup\'erieure, PSL University, CNRS, INRIA, 45 rue d'Ulm, Paris 75005, France}

\date{\today}

\begin{abstract}
The Schr\"odinger wavefunction is ubiquitous in quantum mechanics, quantum chemistry, and bosonic quantum information theory.  
Its zero-set for fermionic systems is well-studied and central for determining chemical properties, yet for bosonic systems the zero-set is less understood, especially in the context of characterizing non-classicality.  
Here we study the zeros of such wavefunctions and give them a novel information-theoretic interpretation. 
Our main technical result is showing that the wavefunction of most bosonic quantum systems can be extended to a holomorphic function over the complex plane, allowing the application of powerful techniques from complex analysis. 
As a consequence, we prove a version of Hudson's theorem for the wavefunction and characterize Gaussian dynamics as classical motion of the wavefunction zeros. 
Our findings suggest that the non-Gaussianity of quantum optical states can be detected by measuring a single quadrature of the electromagnetic field, which we demonstrate in a companion paper \cite{witness}. More generally, our results show that the non-Gaussian features of bosonic quantum systems are encoded in the zeros of their wavefunction.
\end{abstract}

\maketitle


\section{Introduction}

The wavefunction of a physical system provides a fundamental description of its quantum state. Historically, it has led to foundational parallels between quantum theory and wave mechanics \cite{schrodinger1926undulatory,de1992radiations}, and has allowed us to explore post-classical behavior in physics.
Nowadays the wavefunction is still a central object of study, both from a metaphysical point of view \cite{everett1973theory, goldstein2013reality} and a more applied perspective \cite{raymer1997measuring}.
For instance, the understanding of fundamental properties of electrons heavily relies on trial wavefunctions \cite{laughlin1983anomalous, von2005developments}, while the nodes, or \textit{zeros}, of the wavefunction play an important role in the magnetic properties of quantum systems \cite{riess1970nodal}.
Such zeros govern the computational complexity of simulating fermions \cite{ceperley1980ground, ceperley1991fermion, reynolds1982fixed}, in direct relation with the so-called sign problem \cite{bressanini2002we}.

Beyond the Schr\"odinger wavefunction, the zeros of phase-space representations \cite{cahill1969density,rundle2021overview} have been studied notably in the contexts of quantum chaos \cite{leboeuf1990chaos,arranz1996distribution,leboeuf1996universal} and non-Gaussianity of bosonic quantum systems \cite{lutkenhaus1995nonclassical,chabaud_stellar_2020}. In the latter context, the number of complex zeros of the Husimi phase-space representation \cite{husimi1940some} is directly related to a measure of non-Gaussianity known as the stellar rank, which can be understood as a resource for quantum computational tasks \cite{chabaud2023resources}. A similar intuition holds for the Wigner function \cite{wigner1932quantum}: a pure state is non-Gaussian if and only if its Wigner function has zeros \cite{hudson1974wigner,abreu2025inverse}.
Interestingly, the quantum dynamics of phase-space representations sometimes allow an interpretation of the moving zero-set as a many-body classical dynamical system \cite{leboeuf1991phase, chabaud2022holomorphic, bruno2012geometric}.
This is only possible for certain dynamics and when the zeros contain sufficient information about the quantum state. This information is typically associated with a holomorphicity property of the representation \cite{Majorana_1932, bargmann1961hilbert, segal1963mathematical, vourdas2006analytic}.
While for bosons the zeros of phase-space representations provide useful information about the non-classical properties of the quantum system they describe, they do not relate to simple properties of the wavefunction. As such, it is nontrivial to extract the non-classical properties of bosonic quantum systems directly from their wavefunction.
 
Here we address this problem by studying directly the zeros of the bosonic wavefunction. Our contributions are threefold:
\begin{itemize}
    \item \textit{Entire wavefunctions.} We prove that under a simple energy condition, the wavefunction can be extended to an entire function, i.e., a holomorphic function over the whole complex plane (\cref{th:holomorphicwf}). This enables the use of powerful mathematical tools from complex analysis to study the wavefunction.
    \item \textit{Hudson theorem for entire wavefunctions.} As a consequence, we obtain under this energy condition a Hudson theorem for the wavefunction: a pure state is non-Gaussian if and only if its wavefunction has zeros over the complex plane (\cref{th:Hudson}). We further show that the number of zeros of such wavefunctions is equal to the stellar rank. This implies that the wavefunction of a state of finite stellar rank can be written as the product of a polynomial and a Gaussian function (\cref{th:wf_finitestellar}).
    \item \textit{Gaussian dynamics and motion of zeros.} We show that the evolution of the wavefunction under Gaussian Hamiltonians can be seen as a classical many-body dynamical system on the plane, in which the complex zeros of the wavefunction act as classical particles (\cref{th:complexcalogeromoserzeros}). Solving the equations of motion (\cref{th:complexcalogeromoserintegration}), we show that phase-shifts (i.e.\ simple harmonic evolution) typically make the complex zeros of the wavefunction real (\cref{th:wf_have_real_zeros} and \cref{th:odd_stellar_rank_crossing}). This opens the door to detecting non-Gaussianity of bosonic quantum states with a single quadrature measurement, which we explore in a companion paper \cite{witness}.
\end{itemize}

The rest of the paper is structured as follows: after introducing the necessary preliminary material in section~\ref{sec:preliminaries}, we give our main technical result characterizing entire wavefunctions in section~\ref{sec:entire}. Then, we show that the zeros of the wavefunction provide signatures of non-Gaussianity in section~\ref{sec:zeros} by deriving a Hudson theorem for entire wavefunctions and by relating the number of wavefunction zeros to the stellar rank. Finally, in section~\ref{sec:dynamics}, we study the Gaussian dynamics of the wavefunction for finite stellar rank, and relate it to an integrable Calogero--Moser motion of its zeros. We discuss the significance of our results and open questions in section~\ref{sec:conclusion}.


\section{Preliminaries}
\label{sec:preliminaries}

We denote by $\mathcal H = L^2(\mathbb R)$ the Hilbert space of a single bosonic mode, and by $\hat x$, $\hat p$, and $\hat n = \hat a^\dagger \hat a$ the position, momentum, and number operators, respectively, the annihilation operator $\hat{a}$ being defined by $\hat{a} =\frac1{\sqrt2}(\hat{x} + i\hat{p})$. The eigenstates of the number operator form an infinite countable basis of the Hilbert space called the particle-number, photon-number, or Fock basis and it is denoted as $\{\ket n\}_{n\in\mathbb N}$, with $\ket0$ being the vacuum state.

The (position) wavefunction of a state $\ket\psi$ is given by its expansion along the (generalized) eigenvectors of the position operator and denoted as $x\mapsto\psi(x)$. Similarly, the momentum wavefunction is given by the expansion along the (generalized) eigenvectors of the momentum operator, and can be obtained from the position wavefunction by applying the phase-shift operator $\hat R(\theta)=e^{-i\theta\hat n}$ for $\theta=\frac\pi2$. 

Given a normalized single-mode pure state $\ket\psi=\sum_{n\ge0}\psi_n\ket n$, its stellar (or Bargmann \cite{bargmann1961hilbert,segal1963mathematical}) function is defined as $F^\star_\psi(z)\coloneqq\sum_{n\ge0}\frac{\psi_n}{\sqrt{n!}}z^n$ for $z\in\bC$. It is directly related to the Husimi $Q$ function \cite{husimi1940some} by $Q_\psi(z)=\frac{e^{-|z|^2}}\pi|F^\star_\psi(z^*)|^2$. As a result, the stellar function may be thought of as a wavefunction in the overcomplete basis of coherent states, which are the (right) eigenstates of the annihilation operator. The stellar rank is then defined as the (possibly infinite) number of zeros of the stellar function, and any state $\ket{\psi_r}$ of finite stellar rank $r$ can be expressed as \cite{chabaud_stellar_2020}
\begin{equation}\label{eq:parame_rank_r_state}
\ket{\psi_r} =\hat D(\alpha)\hat S(\chi)\sum_{n =0}^r c_n \ket{n},
\end{equation}
where $\hat S(\chi)= e^{\frac{1}{2}(\chi^* \hat a^2 - \chi  \hat a^{\dagger^2})} $ is the squeezing operator with $\chi \in \bC$, $\hat D(\alpha)= e^{\alpha \hat a ^\dagger - \alpha^*\hat a}$ is the displacement operator with $\alpha \in \bC$, and $c_n \in \bC, \forall n$ with $\sum_{n=0}^r |c_n|^2=1$. In that case, the stellar function is the product of a polynomial of degree $r$ and a Gaussian function. In particular, a pure state is a Gaussian state if and only if its stellar rank is equal to zero \cite{lutkenhaus1995nonclassical}, and if and only if its (position or momentum) wavefunction is also a Gaussian function. 
Importantly, the stellar function is an entire function of order at most $2$, i.e., it can be expanded as a Taylor series over the whole complex plane and its growth is bounded by that of a Gaussian function. This has led to a classification of non-Gaussianity based on techniques from complex analysis \cite{chabaud_stellar_2020,chabaud2022holomorphic}.


\section{Entire wavefunctions}
\label{sec:entire}

Hereafter, unless specified otherwise, the wavefunction refers to the position wavefunction.
In general, a valid wavefunction can be any square-integrable function over the reals, that is, an element of $\mathcal H = L^2(\mathbb R)$. We start by a key result which allows to focus our attention on entire functions, i.e., functions that are holomorphic over the whole complex plane.

\begin{theo}[Entire wavefunctions]\label{th:holomorphicwf}
    Let $\ket\psi=\sum_{n\ge0}\psi_n\ket n$ be a pure state. Assume there exists an $s>1$ such that $\ket\psi$ satisfies the energy bound
    \begin{equation}\label{eq:s_energy_bound}
        \langle s^{\hat n}\rangle_\psi=\sum_{n\ge0}s^n|\psi_n|^2<+\infty.
    \end{equation}
    Then the wavefunction of $\ket\psi$ can be extended over the complex plane to an entire function of order at most $2$.
\end{theo}

\noindent We give an informal proof hereafter and refer to Appendix~\ref{app:entire} for the full proof.

\begin{proof}[Proof sketch]
    Let us write such an energy-bounded state as $\ket\psi=\sum_{n=0}^{+\infty}\psi_n\ket n$ in the Fock basis; its wavefunction is given by
    \begin{equation}
        \psi(x)=\frac1{\pi^{1/4}}e^{-x^2/2}\sum_{n=0}^{+\infty}\frac{\psi_n}{\sqrt{2^nn!}}H_n(x),
    \end{equation}
    where $H_n$ is the $n^{\text{th}}$ Hermite polynomial.  We expand $H_n(x)\coloneqq\sum_{k=0}^nh_{n,k}x^k$ and $e^{-x^2/2}=\sum_{l=0}^{+\infty}\frac{(-1)^lx^{2l}}{2^ll!}$ and collect like terms,
    \begin{equation}
        \psi(x)=\sum_{m=0}^{+\infty}\left(\sum_{n=0}^{+\infty}c_{mn}\,\psi_n\right)x^m,
    \end{equation}
    where the energy assumption (\ref{eq:s_energy_bound}) implies an infinite radius of convergence of the resulting series. As a byproduct, extending this function over the complex plane yields the bound
    \begin{equation}\label{eq:bound_entire}
        |\psi(z)|^2\le Ke^{L|z|^2}, \qquad \forall \, z \in \mathbb C
    \end{equation}
    for some constants $K$ and $L$, which implies that the function $z\mapsto\psi(z)$ has an order of growth at most $2$.
\end{proof}

This result allows us to import tools from complex analysis for studying the wavefunctions of non-Gaussian states so long as they satisfy an energy bound of the form of Eq.~(\ref{eq:s_energy_bound}), which is the case for most states. In particular, we show this is the case for the set of pure states with finite stellar rank (see \cref{th:wf_finitestellar}), which is dense in the set of all pure states \cite{chabaud_stellar_2020}. Note that we considered the position wavefunction, but a similar result holds for any phase-shifted wavefunction as the energy bound is rotation-invariant.


\section{Non-Gaussianity and wavefunction zeros}
\label{sec:zeros}

\cref{th:holomorphicwf} has significant complex-analytical consequences, which we discuss in this section. Let us start by mentioning a simple but interesting example: because the zeros of entire functions must be isolated, if the wavefunction of a state $\ket\psi$ vanishes on a continuous segment, then $\ket\psi$ violates the energy bound in Eq.~\eqref{eq:s_energy_bound} for all $s>1$.


\subsection{A Hudson theorem for the wavefunction}

While typically phrased in terms of negative values of the Wigner function, Hudson's theorem can be alternatively stated as identifying a pure state $\ket\psi$ as non-Gaussian if and only if its Wigner function has at least one zero \cite{hudson1974wigner,abreu2025inverse}, or, equivalently, if and only if its Husimi $Q$ function has at least one zero \cite{lutkenhaus1995nonclassical}. Recall that, the zeros of the Husimi $Q$ function correspond (up to complex conjugation and multiplicity) to the zeros of the stellar function, which can be thought of as a wavefunction in the overcomplete basis of coherent states. Since the marginals of the Wigner function yield the probability density functions for the phase-shifted wavefunctions \cite{Bertrand_Bertrand_1987}, it is natural to ask if a similar characterization of non-Gaussianity can be obtained based on the zeros of wavefunctions. This is established by our next result.

\begin{theo}[Hudson theorem for the wavefunction]\label{th:Hudson}
    For a pure state $\ket\psi$, assume there exists an $s>1$ such that $\ket\psi$ satisfies the energy bound $\langle s^{\hat n}\rangle_\psi<+\infty$. Then $\ket\psi$ is non-Gaussian if and only if its wavefunction has a zero over $\mathbb C$.
\end{theo}

\begin{proof}
    If $\ket\psi$ is a Gaussian state, then its wavefunction is a Gaussian function that does not vanish over the complex plane.
    
    Conversely, let us assume that $\ket\psi$ is a state such that $\langle s^{\hat n}\rangle_\psi<+\infty$ for some $s>1$, and that its (extended) wavefunction $z\mapsto\psi(z)$ 
    does not vanish over $\mathbb C$.
    By \cref{th:holomorphicwf}, $z\mapsto\psi(z)$ is an entire function of order at most $2$. And since $z\mapsto\psi(z)$ does not vanish the Hadamard--Weierstrass factorization theorem implies that it is a Gaussian function. Hence, $x\mapsto\psi(x)$ is also a Gaussian function and $\ket\psi$ is a Gaussian state.
\end{proof}

Once again, a similar result holds for the momentum wavefunction (or any other phase-shifted wavefunction) because the energy bound is invariant under phase shifts and phase shifts maps Gaussian states to themselves.  This result establishes the zeros of the Schr\"odinger wavefunction as fundamental signatures of non-Gaussianity.
In what follows, we strengthen this result for states of nonzero but finite stellar rank.


\subsection{Wavefunction zeros for finite stellar rank}
\label{sec:wf_finite_stellar}

The stellar rank quantifies the amount of non-Gaussianity in a quantum state. As discussed in \cref{sec:preliminaries}, it is closely related to the zeros of phase-space representations: a pure state has finite stellar rank if and only if its Husimi $Q$ function has a finite number of zeros \cite{chabaud_stellar_2020}, or equivalently if and only if the zero-set of its Wigner function is bounded \cite{abreu2025inverse}. 
In this section we characterize the stellar rank of non-Gaussian pure quantum states based on the zeros of their wavefunction.

\begin{theo}[Finite stellar rank and wavefunction zeros]\label{th:wf_finitestellar}
    Let $\ket\psi$ be a pure state with stellar rank $r \in \mathbb N$. Then there exists an $s>1$ such that $\langle s^{\hat n}\rangle<+\infty$.  Furthermore, the wavefunction of $\ket{\psi}$ has $r$ complex zeros counted with multiplicity and takes the form 
    \begin{align}
    \psi(x) &= P(x) \, e^{-ax^2 + bx + c},
    \label{eq:general_form_*r_state}
    \end{align}
    for some complex polynomial $P$ of degree $r$, where $a,b,c\in\mathbb{C}$ and $\text{Re}(a) > 0$. 
    
    Conversely, if $\ket\psi$ satisfies $\langle s^{\hat n}\rangle<+\infty$ and $z\mapsto\psi(z)$ has exactly $r$ zeros, then $\ket\psi$ has stellar rank $r$.
\end{theo}

\begin{proof} We defer the proof segment showing the existence of such an $s>1$ to \cref{app:ebound_fsr} (see Lemma~\ref{applem:energySR}).  To show the above form of the wavefunction, note that any pure state $\ket{\psi}$ of finite stellar-rank $r$ may be equivalently expressed as \cite{chabaud_stellar_2020}
\begin{equation}
    \ket{\psi} = P_\psi(\hat a^\dagger) \hat G \ket{0},
\end{equation}
where $\hat{G}$ is a Gaussian unitary (i.e.\ any combination of displacement and squeezing) and $P_\psi(\hat a^\dagger) = \sum_{k=0}^r c_k (\hat a^\dagger)^k$ is a complex polynomial of degree $r$ in the creation operator $\hat a^\dag$ with $c_r \neq 0$.  The state $\hat G \ket{0}$ for any $\hat G$ is by definition a squeezed coherent state \cite{V_V_Dodonov_2002}, the wavefunction of which takes the form of a Gaussian wave packet:
\begin{equation}
    \langle x | \hat G | 0 \rangle \sim e^{- ax^2 + bx + c},
\end{equation}
where $a,b,c\in\mathbb{C}$ such that $\text{Re}(a) > 0$.  Through the identifications $\hat x \mapsto x$ and $\hat p \mapsto -i\partial_x$, the action of the polynomial $P_\psi(\hat a^\dagger)$ on this wavefunction becomes
\begin{equation}
\begin{split}
    \psi(x) &\sim P_\psi\!\left(\frac{x - \partial_x}{\sqrt{2}}\right) e^{- ax^2 + bx + c} \\
    &= \left[ \sum_{k=0}^r c_k (x - \partial_x)^k \right] e^{-ax^2 + bx + c}.
\end{split}
\end{equation}
Each term $(x - \partial_x)^k e^{-ax^2 + bx + c}$ produces a polynomial of degree at most $k$, multiplied by the same Gaussian packet; see for example \cite{blasiak2011combinatorial, wilcox1967exponential} for explicit normal forms.  Summing such terms and collecting powers of $x$, the remaining wavefunction must be of the form
\begin{align}
    \psi(x) &= P(x) \, e^{-ax^2 + bx + c},
\end{align}
for some complex polynomial $ P(x)$ of degree at most $r$, with coefficients depending on $a,b,c$, and $\{c_k\}$. We show now that the polynomial will always be of exactly degree $r$ via induction. For stellar-rank $r=1$ the linear (and hence leading) coefficient of the polynomial is $\frac{c_1}{\sqrt{2}}(1+2a)$ which is always non-zero because $\Re(a)>0$. We denote the polynomial in the wavefunction of a state with stellar-rank $r$ as $P^{(r)}(x)$ and we assume that it has degree $r$ with leading coefficient $p_r \neq 0$. Then direct computation shows that the leading coefficient of $P^{(r+1)}(x)$ is $p_r(1+2a)$ and so is also non-zero. This proves that $P^{(r)}(x)$ is of degree $r $ for all $r$. While this polynomial may not generally have $r$ real zeros, the stellar rank $r$ is nonetheless imprinted into the wavefunction as the number of complex zeros, in a manner analogous to the stellar function.  This remains true along all choices of phase-space rotation angle as such a rotation can be absorbed into the unitary $\hat G$.

For the reverse statement, notice that $\langle s^{\hat n}\rangle<+\infty$ implies by Theorem~\ref{th:holomorphicwf} that for any angle $\theta$ the wavefunction of $\ket{\psi}$ is an entire function of order at most $2$. By the Hadamard--Weierstrass factorization theorem, such a function with exactly $r$ complex zeros (counted with multiplicity) is the product of a polynomial of degree $r$ with a Gaussian function of the form $z\mapsto e^{-az^2+bz+c}$, with $a,b,c\in\mathbb C$. Finally, the fact that the wavefunction is square-integrable over the real line implies that $\Re(a)>0$, which concludes the proof. 
\end{proof}

Note that the energy condition in \cref{th:wf_finitestellar} cannot be omitted: for instance, the position wavefunction $x\mapsto e^{-x^4}$ (up to normalisation) is non-Gaussian but does not vanish. However, it is the wavefunction of a pure state which does not satisfy the energy condition in Eq.~(\ref{eq:s_energy_bound}) for any $s>1$.




\section{Gaussian dynamics and wavefunction zeros}
\label{sec:dynamics}

In this section we study the dynamics of the position wavefunction $\psi(x,t)$ under Gaussian evolution, i.e.\ quantum dynamics generated by Hamiltonians that are polynomials of degree at most $2$ in the position and momentum operators.  Previous work \cite{Mancini_Manko_Tombesi_1995, Mancini_Manko_Tombesi_1996, Mancini_Manko_Tombesi_1997_Classical} has studied the Schr\"odinger evolution of 
marginals of the Wigner function, 
with an emphasis on understanding how this depends on symplectic transformations. 
Here we focus on the extended wavefunction and its zero-set.

\subsection{Gaussian dynamics}

We first show that the entire extension $\psi(z,t)$ of $\psi(x,t)$ follows the same partial differential equation as $\psi(x,t)$, after replacing $\frac{\partial}{\partial x}$ with $\frac{\partial}{\partial z}$.

\begin{lemma}[Schr\"odinger equation over complex space]
    \label{lem:schrodingerextension}
    Let $\ket{\psi}$ be a pure state satisfying the energy condition in \cref{eq:s_energy_bound} and let
    $\hat{H} = P(\hat{x}, \hat{p})$ be any polynomial Hamiltonian.
    Then under the unitary evolution generated by $\hat H$, the extended wavefunction $\psi(z,t)$ satisfies the following Schrödinger equation in complex variables:
    \begin{equation}
        \label{eq:wavefuncevolutioncomplex}
        i \frac{\partial\psi(z,t)}{\partial t} = P\!\left(z, -i\frac{\partial}{\partial z}\right)\psi(z,t).
    \end{equation}
\end{lemma}

\noindent We give a short proof in Appendix~\ref{app:schro}.

When $\ket{\psi}$ is a single-mode quantum state of stellar rank $r$, Theorem~\ref{th:wf_finitestellar} ensures that its extended wavefunction can be expressed as the product of a Gaussian function and a polynomial of degree $r$, with $r$ complex zeros. Moreover, Gaussian unitary evolutions leave the stellar rank invariant \cite{chabaud_stellar_2020} as can be seen from Eq.~\eqref{eq:parame_rank_r_state}. Taken together, it means that we can write: 
\begin{equation}
\psi(z,t) \equiv e^{a(t)z^2 + b(t)z + c(t)}\prod_{k = 1}^r(z-\lambda_k(t)),
\end{equation}
where $a, b, c$ and $\lambda_k, 1\le k \le r$ are complex valued functions.
 In that case, from Eq.~\eqref{eq:wavefuncevolutioncomplex} we derive a differential system satisfied by the $r$ zeros $\lambda_k$ of $\psi(z,t)$ when $\hat H$ is a Gaussian Hamiltonian. This system for the zeros is valid as long as the zeros are distinct. However, as we will see, it is possible to derive from the differential system an expression for $\psi$ which is valid for all $t$, no matter the multiplicity of its zeros. 
\begin{theo}[Motion of zeros of the extended wavefunction]
    \label{th:complexcalogeromoserzeros}
    Consider a single-mode quantum state $\ket{\psi}$ of stellar rank $r$, under a Gaussian evolution generated by the quadratic Hamiltonian: 
    \begin{equation}
    \hat H_G = A\hat x^2 + B \hat p^2 + C \frac{\hat x \hat p + \hat p \hat x}{2} + D \hat x + E \hat p + F.
    \end{equation}
    Let $I$ be an open interval of $\mathbb{R}$, and suppose that for all $t \in I$, $1 \le j, k \le r$: $\lambda_k(t) \neq \lambda_j(t)$.
    Then, for all $t \in I$: 
\begin{equation}
\resizebox{\linewidth}{!}{$
\begin{cases}
\displaystyle \dot a(t)= 4iB a(t)^2 - 2C a(t) - iA, \\[1.2ex]
\displaystyle \dot b(t) = 4iB a(t) b(t) - C b(t) - 2E a(t) - iD, \\[1.2ex]
\displaystyle \dot c(t) = iB \left(b(t)^2 + 6a(t)\right) - \frac{3C}{2} - iD - E b(t) - F, \\[1.2ex]
\displaystyle \ddot \lambda_k(t) = (C^2 - 4AB)\lambda_k(t) + CE - 2BD + 8B^2\sum_{m \neq k} \frac1{(\lambda_k(t) - \lambda_m(t))^3}, \\
\displaystyle \dot \lambda_k(0) = \lambda_k(0)\left(C-4iB a(0)\right) - 2iB b(0) + E + 2iB \sum_{m \neq k} \frac{1}{\lambda_k(0) - \lambda_m(0)},
\end{cases}
$}
\end{equation}
where $\,\dot{}\,$ denotes time derivative.
\end{theo}

\noindent We give a proof in Appendix~\ref{app:dynamics}. Remarkably, under Gaussian evolution the dynamical system decouples the evolution of the zeros and of the Gaussian part of the wavefunction. The former corresponds to the evolution of the non-Gaussian features of the state, while the latter encodes the evolution of its Gaussian features, and the two are only related through their initial conditions. Note that even though Lemma \ref{lem:schrodingerextension} is valid for any polynomial Hamiltonian, Theorem \ref{th:complexcalogeromoserzeros} does not trivially generalize for non-Gaussian evolutions because a general polynomial Hamiltonian will not preserve the stellar rank and so neither the number of zeros of $\psi(z,t)$.

The differential system satisfied by the zeros of the wavefunction described in Theorem \ref{th:complexcalogeromoserzeros} is, up to an affine rescaling to eliminate constant terms, a variant of the classical Calogero--Moser system \cite{calogero,moser1976three} with harmonic confinement (when $C^2 - 4AB < 0$), described and solved in the real-valued case in \cite[Section V]{olshanetsky1981}. The values $\lambda_k$ can be seen as eigenvalues of a Hermitian matrix $X$ satisfying the harmonic equation $\ddot{X} + \omega^2X = 0$ with some specific choice of $\omega$ and initial conditions. This proof naturally extends to the complex case as we show in Appendix \ref{app:complexcalogeromoserintegration}. As mentioned earlier, this harmonic equation is valid even without the hypothesis that the zeros are distinct in an open time interval around $t$. We simply need that the zeros are \emph{initially} distinct, which we can always suppose up to a perturbation of the initial state.

\begin{theo}[Resolution of the dynamical system for the zeros]
    \label{th:complexcalogeromoserintegration}

    With the same notation as in \cref{th:complexcalogeromoserzeros}, define, for all $1 \le k \le r$: 

    \begin{equation}
        \tilde{\lambda}_k \coloneq(4B^2)^{-4}\left(\lambda_k - \frac{CE - 2BD}{\omega^2}\right),
    \end{equation}
    where $\omega^2\coloneqq4AB - C^2$. Suppose that for all $1 \le j, k \le r: \lambda_k(0) \neq \lambda_j(0)$.
     Then, the values $\tilde{\lambda}_k(t)$ are exactly the eigenvalues of the matrix: 
    \begin{equation}
    \tilde X(t) = \tilde \Lambda e^{i\omega t} + \frac{L}{\omega}\sin(\omega t),
    \end{equation} 
    where we defined the $r \times r$ complex matrix $L$ with entries: 
    \begin{equation}
    L_{jk} = (\dot {\tilde \lambda}_j(0) + i \omega \tilde \lambda_j(0))\delta_{jk} + i (1-\delta_{jk})\frac{1}{\tilde \lambda_j(0) - \tilde \lambda_k(0)},
    \end{equation}
    and $\tilde \Lambda = \mathrm{diag}(\tilde \lambda_1(0), \ldots, \tilde \lambda_r(0)).$
    
    An immediate corollary is a concise formula for the extended wavefunction of a single-mode quantum state of finite stellar rank $r$ evolving under any Gaussian unitary:
    \begin{equation}
        \psi(z, t) = e^{a(t)z^2 + b(t)z + c(t)}\det(X(t) - zI),
    \end{equation}
    where $X(t) \equiv (4B^2)^4 \tilde X(t) + \frac{CE-2BD}{\omega^2}I$.
\end{theo}

Note that we can derive $L$ using the first order conditions of Theorem \ref{th:complexcalogeromoserzeros}.

\subsection{Phase shifts}

We now focus on specific Gaussian evolutions, namely phase shifts, where the Hamiltonian is proportional to the number operator, and we study sufficient conditions for the wavefunction to have real-valued zeros at some points during the evolution.

The following theorem, derived from the matrix resolution of the Calogero--Moser system in Theorem \ref{th:complexcalogeromoserintegration}, guarantees that there always exists a phase-shift angle such that the corresponding wavefunction of the target state has a zero on the real line under the assumption that $\psi(z, 0)$ has sufficiently far apart zeros.

\begin{theo}[Existence of real zeros for phase-shifted wavefunctions under distance condition]
\label{th:wf_have_real_zeros}
 With the same notations and assumptions as in \cref{th:complexcalogeromoserzeros}, suppose that:
\begin{equation}
\min_{i, j} \vert \lambda_i(0) - \lambda_j(0)\vert \ge \sqrt{\frac{r-1}{a(0)}},
\end{equation}
where $r$ is the stellar rank of $\ket{\psi}$.
Then, for all $j$ there exists $t_j<t_j'$ such that $\lambda_j(t_j),\lambda_j(t_j')\in\mathbb R$. In particular the sum of the number of real zeros over all phase-shifted wavefunctions of $\ket{\psi}$ is larger or equal to $2r$.
\end{theo}
The proof, given in \cref{app:isolated}, is based on treating the Calogero--Moser interaction term in the equation of Theorem \ref{th:complexcalogeromoserzeros} as a perturbation of a simple harmonic movement, which is possible only under the hypothesis that the zeros are sufficiently far apart at $t=0$. However, we conjecture that for a pure stellar rank $r$, each of the complex zeros of its wavefunction becomes real under at least one suitable phase shift of the state. This conjecture comes from numerical evidence (we have not found a counter-example) and from a weaker version of the result valid for a very large family of states, as the following result illustrates.

\begin{theo}[Existence of real zeros for phase-shifted wavefunctions under imbalance condition]
    \label{th:odd_stellar_rank_crossing}
    Let $\ket{\psi}$ be a state of finite stellar rank, with the complex zeros of its wavefunction denoted as $\{\lambda_i\}$. Let also $n^+$ (resp.\ $n^-$) be the number of indices $i$ such that $\Im(\lambda_i) > 0$ (resp.\ $\Im(\lambda_i) < 0$). Suppose that $n^+ \neq n^-$, and that all zeros are simple at $t=0$. Then, at least one phase-shifted wavefunction of $\ket{\psi}$ has a real-valued zero.
\end{theo}

\noindent We give a proof in \cref{app:oddcross}.
In particular, \cref{th:odd_stellar_rank_crossing} directly implies that the result holds for any state of finite, odd stellar rank (with simple zeros). 

We also conjecture that Theorem \ref{th:wf_have_real_zeros} may be refined in the following way: if the initial zeros are pairwise far apart enough, then the number of real zeros over all phase-shifted wavefunctions of $\psi$ is \emph{exactly} $2r$. This comes from the observation that under this condition, the trajectory of the zeros is very close from an ellipse centered on the origin, but we were not able to exclude, for instance, a "whirling" motion of the zeros around the ellipse, that would entail multiple crossings of the real axis in a small time window. If this conjecture 
is true it would open the possibility to witness not only non-Gaussianity, but also stellar rank, solely from homodyne measurements and squeezing.

These results uncover strong connections between the zeros of phase-shifted wavefunctions (or equivalently, of probability density functions for position, momentum, and phase-shifted versions thereof) and non-Gaussianity, as measured by the stellar rank. In quantum optics, probability density functions for phase-shifted position-like operators (known as quadrature operators) can be directly measured using homodyne detection \cite{leonhardt2010essential}. This suggests that it may be possible to exploit low probabilities of few quadrature measurement outcomes for witnessing non-Gaussianity, signaled by the real-valued zeros of the phase-shifted wavefunctions. We demonstrate this fact in a companion paper \cite{witness}, in which we introduce and analyze a quantum optical protocol based on zeros of the wavefunction allowing to witness non-Gaussianity and stellar rank using a single quadrature.


\section{Conclusion}
\label{sec:conclusion}

By extending the Schr\"odinger wavefunction to a holomorphic function on the complex plane, we have given a novel interpretation of its zero-set for a wide class of quantum states and established a direct connection to that of stellar rank \cite{chabaud_stellar_2020}.  This has led to a generalization of Hudson's theorem and a characterization of Gaussian dynamics as a Calogero-Moser (i.e.\ classical) evolution of these zeros.  We have also described a handful of circumstances whereby these complex zeroes become real along phase-shifting evolution, i.e.\ along simple harmonic motion.

There are several possible directions for future research.  One would be to either prove, refine, or reject our conjecture that for any pure state with finite stellar rank $r$, each complex zero in the extended wavefunction becomes real within at least one rotation angle.  This is mathematically equivalent to studying how the extended wavefunction behaves under the fractional Fourier transform.  If true, this would suggest that for each energy-bounded \eqref{eq:s_energy_bound} state $\ket{\psi}$ there is a collection of privileged angles along which each zero becomes observable in the corresponding quadrature measurement; see our companion paper \cite{witness} for first steps in this direction. 

Another avenue to pursue is generalizing our framework to multiple canonical degrees of freedom, i.e.\ to $L^2 (\mathbb R^n)$.  Similar to stellar rank, this would likely entail a generalization from countable zero-sets over $\mathbb C$ to continuous zero-sets over $\mathbb C^n$.  This may or may not be straightforward in the context of wavefunctions, and it is worth noting that a different proof technique is used to establish the many-body version of Hudson's theorem \cite{soto1983wigner}.  Finally, a generalization to mixed states is needed; one approach could be to use a convex roof construction, again analogous to stellar rank.


\section*{Acknowledgements}

UC thanks Z.~Van Herstraeten, J.~Prata and N.~Dias for interesting discussions. SC acknowledges E. Moulinier for his insightful remarks. UC and JD acknowledge funding from the European Union’s Horizon Europe Framework Programme (EIC Pathfinder Challenge project Veriqub) under Grant Agreement No.~101114899.


\bibliography{ref}


\onecolumngrid

\appendix

\newpage

\begin{center}
    {\huge Appendix}
\end{center}


\section{Hudson's theorem for the wavefunction (proof of Theorem~\ref{th:holomorphicwf})}
\label{app:entire}

In this section we complete the proof of Theorem~\ref{th:holomorphicwf}. From the proof sketch given in the main text, it is enough to prove the following result:

\begin{lemma}\label{applem:wf_over_complex}
    Let $\ket\psi$ be a single-mode state satisfying the energy bound $\langle s^{\hat n}\rangle_\psi<+\infty$ for some $s>1$. Then, its position wavefunction $\psi$ can be extended to an entire function over the complex plane of order at most $2$.
\end{lemma}

\begin{proof}
    Let $N\in\mathbb N$ and let $\ket\psi=\sum_{n=0}^{+\infty}\psi_n\ket n$, with $\langle s^{\hat n}\rangle_\psi=\sum_{n=0}^{+\infty}s^n|\psi_n|^2<+\infty$ for some $s>1$. The position wavefunction of $\ket\psi$ is given by
    \begin{equation}
        \psi(x)=\frac1{\pi^{1/4}}e^{-x^2/2}\sum_{n=0}^{+\infty}\frac{\psi_n}{\sqrt{2^nn!}}H_n(x),
    \end{equation}
    for all $x\in\mathbb R$, where the Hermite polynomials are given by
    \begin{equation}\label{eq:Hermiteevenodd}
        H_n(x)=\begin{cases}n!\sum_{l=0}^{\frac n2}\frac{(-1)^{\frac n2-l}2^{2l}}{(2l)!(\frac n2-l)!}x^{2l}&n\text{ even},\\\\n!\sum_{l=0}^{\frac{n-1}2}\frac{(-1)^{\frac{n-1}2-l}2^{2l+1}}{(2l+1)!(\frac{n-1}2-l)!}x^{2l+1}&n\text{ odd}.\end{cases}
    \end{equation}
    We write $H_n(x)=\sum_{k=0}^nh_{n,k}x^k$ for brevity, so that
    \begin{equation}\label{eq:splitevenodd}
        \begin{aligned}
        \pi^{1/4}\psi(x)&=e^{-x^2/2}\sum_{n=0}^{+\infty}\frac{\psi_n}{\sqrt{2^nn!}}\sum_{k=0}^nh_{n,k}x^k\\
        &=\sum_{m=0}^{+\infty}\sum_{n=0}^{+\infty}\sum_{k=0}^n\frac{\psi_n}{\sqrt{2^nn!}}\frac{(-1)^m}{2^mm!}h_{n,k}x^{2m+k}\\
        &=\sum_{m=0}^{+\infty}\sum_{p=0}^{+\infty}\sum_{l=0}^p\frac{\psi_{2p}}{2^p\sqrt{(2p)!}}\frac{(-1)^m}{2^mm!}h_{2p,2l}x^{2m+2l}+\sum_{m=0}^{+\infty}\sum_{p=0}^{+\infty}\sum_{l=0}^p\frac{\psi_{2p+1}}{2^p\sqrt{2(2p+1)!}}\frac{(-1)^m}{2^mm!}h_{2p+1,2l+1}x^{2m+2l+1},
        \end{aligned}
    \end{equation}
    where in the last line we have split the sums for $n$ even and odd. We start with the first sum. With Eq.~(\ref{eq:Hermiteevenodd}) it reads as follows:
    \begin{align}
        f_\mathrm{even}(x)\coloneqq&\sum_{m=0}^{+\infty}\sum_{p=0}^{+\infty}\sum_{l=0}^p\frac{\psi_{2p}}{2^p\sqrt{(2p)!}}\frac{(-1)^m}{2^mm!}h_{2p,2l}x^{2m+2l}\\
        =&\sum_{m=0}^{+\infty}\sum_{p=0}^{+\infty}\sum_{l=0}^p\frac{\psi_{2p}}{2^p\sqrt{(2p)!}}\frac{(-1)^m}{2^mm!}\frac{(2p)!(-1)^{p-l}2^{2l}}{(2l)!(p-l)!}(x^2)^{m+l}.
    \end{align}
    We assume hereafter that the sums can be freely rearranged under the condition $\langle s^{\hat n}\rangle_\psi<+\infty$. The derivation will show that the assumption is justified. Setting $q=m+l$ we obtain
    \begin{equation}\label{eq:psi_even_inter}
        f_\mathrm{even}(x)=\sum_{q=0}^{+\infty}(-x^2/2)^q\left(\sum_{p=0}^{+\infty}\frac{(-1)^p\psi_{2p}\sqrt{(2p)!}}{2^p}\sum_{l=0}^{\min(p,q)}\frac{2^{3l}}{(2l)!(p-l)!(q-l)!}\right).
    \end{equation}
    Now for all $p,q\in\mathbb N$, and for all $t>1$ we have:
    \begin{align}
        \sum_{l=0}^{\min(p,q)}\frac{2^{3l}}{(2l)!(p-l)!(q-l)!}&=\frac1{p!q!}\sum_{l=0}^{\min(p,q)}8^l\frac{\binom pl\binom ql}{\binom{2l}l}\\
        &\le\frac1{p!q!}\sum_{l=0}^{\min(p,q)}8^l\binom pl\binom ql\\
        &=\frac1{p!q!}\sum_{l=0}^{\min(p,q)}(t-1)^l\binom pl\left(\frac8{t-1}\right)^l\binom ql\\
        &\le\frac1{p!q!}\left(1+\frac8{t-1}\right)^q\sum_{l=0}^{\min(p,q)}(t-1)^l\binom pl\\
        &\le\frac1{p!q!}t^p\left(1+\frac8{t-1}\right)^q,
    \end{align}
    where the last two steps follow from the binomial theorem. With Eq.~(\ref{eq:psi_even_inter}) this yields
    \begin{equation}\label{eq:boundeven}
    \begin{aligned}
        |f_\mathrm{even}(x)|&\le\left(\sum_{p=0}^{+\infty}\frac{t^p|\psi_{2p}|\sqrt{(2p)!}}{2^pp!}\right)\sum_{q=0}^{+\infty}\left(\frac12+\frac4{t-1}\right)^q\frac{x^{2q}}{q!}\\
        &=\left(\sum_{p=0}^{+\infty}\frac{t^p|\psi_{2p}|\sqrt{(2p)!}}{2^pp!}\right)\exp\left[\left(\frac12+\frac4{t-1}\right)x^2\right]\\
        &\le\left(\sum_{p=0}^{+\infty}t^p|\psi_{2p}|\right)\exp\left[\left(\frac12+\frac4{t-1}\right)x^2\right],
    \end{aligned}
    \end{equation}
    where we used $\frac{(2p)!}{2^{2p}(p!)^2}\le1$ in the last line.
    
    Similarly, for the odd part we have
    \begin{align}
        f_\mathrm{odd}(x)\coloneqq&\sum_{m=0}^{+\infty}\sum_{p=0}^{+\infty}\sum_{l=0}^p\frac{\psi_{2p+1}}{2^p\sqrt{2(2p+1)!}}\frac{(-1)^m}{2^mm!}h_{2p+1,2l+1}x^{2m+2l+1}\\
        =&\,x\sum_{m=0}^{+\infty}\sum_{p=0}^{+\infty}\sum_{l=0}^p\frac{\psi_{2p+1}}{2^p\sqrt{2(2p+1)!}}\frac{(-1)^m}{2^mm!}\frac{(2p+1)!(-1)^{p-l}2^{2l+1}}{(2l+1)!(p-l)!}(x^2)^{m+l}.
    \end{align}
    Once again, rearranging the sums and setting $q=m+l$ we obtain
    \begin{equation}\label{eq:psi_odd_inter}
        f_\mathrm{odd}(x)=\sqrt2x\sum_{q=0}^{+\infty}(-x^2/2)^q\left(\sum_{p=0}^{+\infty}\frac{(-1)^p\psi_{2p+1}\sqrt{(2p+1)!}}{2^p}\sum_{l=0}^{\min(p,q)}\frac{2^{3l}}{(2l+1)!(p-l)!(q-l)!}\right).
    \end{equation}
    Now for all $p,q\in\mathbb N$, and for all $t>1$ we have:
    \begin{align}
        \sum_{l=0}^{\min(p,q)}\frac{2^{3l}}{(2l+1)!(p-l)!(q-l)!}&=\frac1{p!q!}\sum_{l=0}^{\min(p,q)}\frac{8^l}{2l+1}\frac{\binom pl\binom ql}{\binom{2l}l}\\
        &\le\frac1{p!q!}\sum_{l=0}^{\min(p,q)}2^l\binom pl\binom ql\\
        &=\frac1{p!q!}\sum_{l=0}^{\min(p,q)}(t-1)^l\binom pl\left(\frac2{t-1}\right)^l\binom ql\\
        &\le\frac1{p!q!}\left(1+\frac2{t-1}\right)^q\sum_{l=0}^{\min(p,q)}(t-1)^l\binom pl\\
        &\le\frac1{p!q!}t^p\left(1+\frac2{t-1}\right)^q,
    \end{align}
    where we used $(2l+1)\binom{2l}l\ge4^l$ in the second line and where the last two steps follow from the binomial theorem. With Eq.~(\ref{eq:psi_odd_inter}) this yields
    \begin{equation}\label{eq:boundodd}
    \begin{aligned}
        |f_\mathrm{odd}(x)|&\le\left(\sum_{p=0}^{+\infty}\frac{t^p|\psi_{2p+1}|\sqrt{(2p+1)!}}{2^pp!}\right)\sqrt2|x|\exp\left[\left(\frac12+\frac1{t-1}\right)x^2\right]\\
        &\le\left(\sum_{p=0}^{+\infty}\frac{t^p|\psi_{2p+1}|\sqrt{(2p+1)!}}{2^pp!}\right)\exp\left[\left(\frac12+\frac1e+\frac1{t-1}\right)x^2\right]\\
        &\le\left(\sum_{p=0}^{+\infty}t^p\sqrt{2p+1}|\psi_{2p+1}|\right)\exp\left[\left(\frac12+\frac1e+\frac1{t-1}\right)x^2\right],
    \end{aligned}
    \end{equation}
    where we used $\sqrt2|x|\le e^{x^2/e}$ in the second line and $\frac{(2p)!}{2^{2p}(p!)^2}\le1$ in the last line.
    
    Finally, given $s>1$ we pick $t=s^\alpha>1$, where $0\le\alpha<\frac12$. We have $t^p\le s^{p/2}$ and $t^p\sqrt{2p+1}=o(s^{p/2})$ so there exists $C\ge1$ which depends only on $\alpha$ and $s$ such that $t^p\sqrt{2p+1}\le Cs^{(2p+1)/4}$ for all $p\in\mathbb N$. Combined with Eqs.~(\ref{eq:splitevenodd}), (\ref{eq:boundeven}), and (\ref{eq:boundodd}) we obtain
    \begin{align}
        \pi^{1/4}|\psi(x)|&\le|f_\mathrm{even}(x)|+|f_\mathrm{odd}(x)|\\
        &\le\left(\sum_{p=0}^{+\infty}t^p|\psi_{2p}|\right)\exp\left[\left(\frac12+\frac4{s^\alpha-1}\right)x^2\right]+\left(\sum_{p=0}^{+\infty}t^p\sqrt{2p+1}|\psi_{2p+1}|\right)\exp\left[\left(\frac12+\frac1e+\frac1{s^\alpha-1}\right)x^2\right]\\
        &\le\left(\sum_{p=0}^{+\infty}s^{(2p)/4}|\psi_{2p}|\right)\exp\left[\left(\frac12+\frac4{s^\alpha-1}\right)x^2\right]+C\left(\sum_{p=0}^{+\infty}s^{(2p+1)/4}|\psi_{2p+1}|\right)\exp\left[\left(\frac12+\frac1e+\frac1{s^\alpha-1}\right)x^2\right]\\
        &\le C\left(\sum_{n=0}^{+\infty}s^{n/4}|\psi_n|\right)\exp\left[\left(\frac12+\frac1e+\frac4{s^\alpha-1}\right)x^2\right]\\
        &\le C\sqrt{\sum_{n=0}^{+\infty}\left(\frac1{\sqrt s}\right)^n\left(\sum_{n=0}^{+\infty}s^n|\psi_n|^2\right)}\exp\left[\left(\frac12+\frac1e+\frac4{s^\alpha-1}\right)x^2\right],
    \end{align}
    where we used Cauchy--Schwarz inequality in the last line.
    Hence,
    \begin{equation}
         |\psi(x)|^2\le Ke^{Lx^2},
    \end{equation}
    where we have set
    \begin{align}
        K=\frac{C^2\sqrt s}{(\sqrt s-1)\sqrt\pi}\langle s^{\hat n}\rangle_\psi,
    \end{align}
    and
    \begin{align}
        L=1+\frac2e+\frac8{s^\alpha-1}.
    \end{align}
    The same derivation with truncated sums shows that all series have infinite convergence radius and that the infinite sums can be freely rearranged under the condition $\langle s^{\hat n}\rangle_\psi<+\infty$. In particular, we obtain for all $z\in\mathbb C$:
    \begin{equation}
         |\psi(z)|^2\le Ke^{L|z|^2},
    \end{equation}
    which shows that $\psi$ is of order at most $2$.
\end{proof}


\section{Energy bounds for states of finite stellar rank}
\label{app:ebound_fsr}

In this section, we prove that states of finite stellar rank satisfy the energy bound in Eq.~\eqref{eq:s_energy_bound}, which is used in the proof of \cref{th:wf_finitestellar} in the main text.

\begin{lemma}[Energy bound for finite stellar rank]\label{applem:energySR}
    Let $\ket\psi$ be a state of finite stellar rank. Then, there exists $s>1$ such that $\langle s^{\hat n}\rangle_\psi<+\infty$.
\end{lemma}

\begin{proof}
    
To prove this result, we combine three intermediate results:

\begin{lemma}[Energy bound for squeezed vacuum states]\label{applem:energysq}
    Let $\ket\xi=\hat S(\xi)\ket0$ be a squeezed vacuum state for $\xi\in\mathbb C$. Then, for all $s>1$ such that $s<\frac1{\tanh|\xi|}$ we have $\langle s^{\hat n}\rangle_\xi<+\infty$.
\end{lemma}

\begin{lemma}[Energy bound under photon addition]\label{applem:energyadd}
    Let $s>1$ and let $\ket\psi$ be a pure state such that $\langle s^{\hat n}\rangle_\psi<+\infty$. Let $\ket{\phi}:=\hat a^\dag\ket\psi$. Then, for all $t<s$, $\langle t^{\hat n}\rangle_\phi<+\infty$.
\end{lemma}

\begin{lemma}[Energy bound under displacement]\label{applem:energydisp}
    Let $s>1$ and let $\ket\psi$ be a pure state such that $\langle s^{\hat n}\rangle_\psi<+\infty$. Let $\alpha\in\mathbb C$ and let $\ket{\psi(\alpha)}:=\hat D(\alpha)\ket\psi$. Then, for all $t<s$, $\langle t^{\hat n}\rangle_{\psi(\alpha)}<+\infty$.
\end{lemma}

\noindent The proof of \cref{applem:energySR} is then concluded by noticing that (up to normalisation) any state of finite stellar rank can be obtained from a squeezed vacuum state by applying a finite number of displacement operators and creation operators \cite[Eq.\ (5)]{chabaud_stellar_2020}.
\end{proof}

\begin{proof}[Proof of \cref{applem:energysq}]
Let $\xi\in\mathbb C$. Writing $\hat S(\xi)=e^{\frac12(\xi^*\hat a^2-\xi\hat a^{\dag2})}$ the squeezing operator with $\xi:=re^{i\phi}$, the squeezed vacuum state with squeezing parameter $\xi$ can be expanded in Fock basis as
\begin{equation}
    \ket\xi=\frac1{\sqrt{\cosh r}}\sum_{n\ge0}(-e^{i\phi}\tanh r)^n\frac{\sqrt{(2n)!}}{2^nn!}\ket{2n}.
\end{equation}
Therefore,
\begin{align}
    \langle s^{\hat n}\rangle_\xi&=\sum_{n\ge0}s^n|\langle n|\xi\rangle|^2\\
    &=\frac1{\cosh r}\sum_{n\ge0}s^{2n}(\tanh r)^{2n}\frac{(2n)!}{4^n(n!)^2}\\
    &\le\frac1{\cosh r}\sum_{n\ge0}s^{2n}(\tanh r)^{2n},
\end{align}
where we used $\frac{(2n)!}{4^n(n!)^2}\le1$ (which follows from the binomial theorem) and where the series in the last line converges whenever $s<\frac1{\tanh|\xi|}$. Note that the latter bound is tight as can be seen by using $\frac{(2n)!}{4^n(n!)^2}\ge\frac1{2n+1}$.
\end{proof}

\begin{proof}[Proof of \cref{applem:energyadd}]
With the notations of the lemma, we write $\ket\psi=\sum_{n\ge0}\psi_n\ket n$, so that 
\begin{equation}
    \ket\phi=\sum_{n\ge0}\psi_n\sqrt{n+1}\ket{n+1}.
\end{equation}
Then, for all $t<s$,
\begin{align}
    \langle t^{\hat n}\rangle_\phi&=\sum_{n\ge0}t^n|\langle n|\phi\rangle|^2\\
    &=\sum_{n\ge0}t^n(n+1)|\psi_n|^2,
\end{align}
and $t^n(n+1)|\psi_n|^2=o(s^n|\psi_n|^2)$, so the series in the last line converges.
\end{proof}

\begin{proof}[Proof of \cref{applem:energydisp}]
We first compute the sum of a series defined using Fock state amplitudes of a displacement operator. Defining the 2D Laguerre polynomials as \cite[Eq.\ (1.2)]{wunsche2016generating}
\begin{equation}
    L_{m,n}(z,z'):=\sum_{j=0}^{\min(m,n)}\frac{(-1)^jm!n!}{j!(m-j)!(n-j)!}z^{m-j}w^{n-j},
\end{equation}
we have \cite[Eqs.\ (3.2) and (9.4)]{wunsche1998laguerre}
\begin{equation}
    \langle m|\hat D(\alpha)|n\rangle=\frac{(-1)^ne^{-\frac12|\alpha|^2}}{\sqrt{m!n!}}L_{m,n}(\alpha,\alpha^*),
\end{equation}
and thus
\begin{align}
    |\langle m|\hat D(\alpha)|n\rangle|^2&=\frac{e^{-|\alpha|^2}}{m!n!}|L_{m,n}(\alpha,\alpha^*)|^2\\
    &=\frac{e^{-|\alpha|^2}}{m!n!}[L_{m,n}(|\alpha|,|\alpha|)]^2\\
    &=\frac{e^{-|\alpha|^2}}{m!n!}L_{m,n}(|\alpha|,|\alpha|)L_{n,m}(|\alpha|,|\alpha|).
\end{align}
Using the generating function for products of 2D Laguerre polynomials \cite[Eq.\ (8.4)]{wunsche2016generating} we obtain, for all $s>1$,
\begin{align}
    \sum_{n\ge0}\frac1{s^n}|\langle m|\hat D(\alpha)|n\rangle|^2&=\frac{e^{-|\alpha|^2}}{m!}\sum_{n\ge0}\frac{(1/s)^n}{n!}L_{m,n}(|\alpha|,|\alpha|)L_{n,m}(|\alpha|,|\alpha|)\\
    &=\frac{e^{-|\alpha|^2}}{m!}e^{\frac1s|\alpha|^2}\frac1{s^m}(-1)^mL_{m,m}(|\alpha|(1-s),|\alpha|(1-s))\\
    &=\frac{e^{-(1-\frac1s)|\alpha|^2}}{s^m}L_m((s-1)^2|\alpha|^2),
\end{align}
where $L_m$ is the $m^{th}$ Laguerre polynomial and where we used \cite[Eq.\ (1.7)]{wunsche2016generating} in the last line. Using the generating function for the Laguerre polynomials we finally obtain, for all $t<s$,
\begin{equation}\label{eq:dispseries} 
\begin{aligned}
    \sum_{m\ge0}t^m\sum_{n\ge0}\frac1{s^n}|\langle m|\hat D(\alpha)|n\rangle|^2&=e^{-(1-\frac1s)|\alpha|^2}\sum_{m\ge0}(t/s)^mL_m((s-1)^2|\alpha|^2)\\
    &=\frac s{s-t}e^{-(s-1)|\alpha|^2(\frac1s+\frac{t(s-1)}{s-t})}.
\end{aligned}
\end{equation}
Now with the notations of the lemma, we write $\ket\psi=\sum_{n\ge0}\psi_n\ket n$, so that for all $t<s$,
\begin{align}
    \langle t^{\hat n}\rangle_{\psi(\alpha)}&=\sum_{m\ge0}t^m|\langle m|\psi(\alpha)\rangle|^2\\
    &=\sum_{m\ge0}t^m|\langle m|\hat D(\alpha)|\psi\rangle|^2\\
    &=\sum_{m\ge0}t^m\left|\sum_{n\ge0}s^{n/2}\psi_n\,s^{-n/2}\langle m|\hat D(\alpha)|n\rangle\right|^2\\
    &\le\sum_{m\ge0}t^m\sum_{n\ge0}\frac1{s^n}|\langle m|\hat D(\alpha)|n\rangle|^2\left(\sum_{n\ge0}s^n|\psi_n|^2\right)\\
    &=\frac s{s-t}e^{-(s-1)|\alpha|^2(\frac1s+\frac{t(s-1)}{s-t})}\langle s^{\hat n}\rangle_\psi,
\end{align}
where we used the Cauchy--Schwarz inequality in the fourth line and Eq.~(\ref{eq:dispseries}) in the last line. This shows that the series converges for all $t<s$.

\end{proof}


\section{Schr\"odinger equation over complex space (proof of Lemma \ref{lem:schrodingerextension})}
\label{app:schro}

By definition of the momentum and position operator in wavefunction space, we know that if $\hat{H} = P(\hat{p}, \hat{q})$, then $\psi(x,t)$ obeys the Schrödinger equation: 
\begin{equation}
\label{eq:schrodpolreal}
i\frac{\partial \psi(x, t)}{\partial t} = P\!\left(x, -i\frac{\partial}{\partial x}\right)\psi(x,t).
\end{equation}

To extend the validity of this equation to the full complex plane, consider the series expansion of $\psi(z, t)$, which is well-defined by \cref{th:holomorphicwf}:
\begin{equation}
\psi(z, t) = \sum_{n=0}^{+\infty}c_n(t)z^n.
\end{equation}
The evolution in Eq.~(\ref{eq:schrodpolreal}) is equivalent to a differential system involving the functions $c_n$ . Indeed, we can expand $P$ as: 
\begin{equation}
P\!\left(x, -i\frac{\partial}{\partial x}\right) = \sum_{k, j \ge 0}d_{kj}x^k\frac{\partial^j}{\partial x^j},
\end{equation}
where $d_{kj} \in \mathbb{C}$ is nonzero only for a finite number of indices $(k, j) \in \mathbb{N}^2$. 
We thereby obtain a series expansion for the right hand side of \ref{eq:schrodpolreal} :

\begin{align}
    P\!\left(x, -i\frac{\partial}{\partial x}\right)\psi(x,t) &= \sum_{k, j, m \ge 0}d_{kj}x^k\frac{\partial^j}{\partial x^j}c_m(t)x^m\\
    &= \sum_{m \ge 0}\sum_{k \ge 0}\sum_{j = 0}^n (m)_j d_{k,j}c_m(t)x^{m-j+k}\\
    &= \sum_{n \ge 0} \left(\sum_{\substack{0 \le j \le m \\ 0 \le k \\ m-j+k = n}}(m)_j d_{k,j}c_m(t)\right)x^n,
\end{align}
where we introduced the falling factorial notation $(m)_j = \prod_{i = 0}^{j-1}(m-i)$.

The left hand side of (\ref{eq:schrodpolreal}) expands as:

\begin{equation}
i\frac{\partial \psi(x,t)}{\partial t}=\sum_{n}i\frac{dc_n(t)}{dt}x^n,
\end{equation}
which, by identification of the coefficients in the series expansion, yields the differential system:

\begin{equation}
\forall n \in \mathbb{N}, i\frac{dc_n}{dt}= \sum_{\substack{ 
0 \le j \le m \\ 0 \le k \\ m-j+k = n}}(m)_jd_{k,j}c_m.
\end{equation}
Replacing the real variable $x$ by the complex variable $z$, and the operator $\frac{\partial}{\partial x}$ by $\frac{\partial}{\partial z}$ in the previous computation, we note that Eq.~(\ref{eq:wavefuncevolutioncomplex}) is equivalent to the same differential system, which concludes the proof.

Note that this holomorphic extension of the wavefunction to the complex plane does not follow from the analytic continuation theorem. Indeed, using this theorem would require the original function to be defined and to be analytic on a connected open subset of $\mathbb{C}$. Here, the wavefunction is originally only defined on the real axis, which is of course not open in $\mathbb{C}$.


\section{Motion of zeros of the wavefunction (proof of Theorem \ref{th:complexcalogeromoserzeros})}
\label{app:dynamics}

The proof is similar to that of \cite[Theorem 1]{chabaud2022holomorphic}.

Let us examine the effect of $\frac{\partial}{\partial t}$ $\frac{\partial}{\partial z}$ and $\frac{\partial^2}{\partial z^2}$ on the wavefunction. Hereafter, $\,\dot{}\,$ is used to denote time derivative.
\begin{align}
    \frac{\partial \psi(z,t)}{\partial t} &= \left(\dot a(t)z^2 + \dot b(t)z + \dot c(t)\right)\psi(z,t) - e^{a(t)z^2 + b(t)z + c(t)}\sum_{k = 1}^r\dot \lambda_k(t) \prod_{j \neq k}(z - \lambda_j(t))\\
    \frac{\partial \psi(z,t)}{\partial z} &= (2a(t)z + b(t))\psi(z,t)+e^{a(t)z^2 + b(t)z + c(t)}\sum_{k = 1}^r\prod_{j\neq k}(z - \lambda_j(t))\\
    \frac{\partial^2\psi(z,t)}{\partial z^2} &= \frac{\partial^2}{\partial z^2}\left(e^{a(t)z^2 + b(t)z + c(t)}\right)\prod_{k=1}^r(z-\lambda_k(t)) + 2 \frac{\partial}{\partial z}\left(e^{a(t)z^2 + b(t)z + c(t)} \right)\frac{\partial}{\partial z} \left( \prod_{k =1}^r(z-\lambda_k(t)) \right) \\
    &\quad\quad+ e^{a(t)z^2 + b(t)z + c(t)}\frac{\partial^2}{\partial z^2}\left(\prod_{k = 1}^r(z-\lambda_k(t))\right) \\
    &=\left((2a(t)z + b(t))^2 + 2a(t)\right)\psi(z,t) +\left(4a(t)z + 2b(t)\right)e^{a(t)z^2 + b(t)z + c(t)}\sum_{k = 1}^r\prod_{j \neq k}(z - \lambda_j(t)) \\
    &\quad\quad+ e^{a(t)z^2 + b(t)z + c(t)}\sum_{1 \le k, m \le r}\prod_{j \neq k, m}(z - \lambda_j(t)).
\end{align}

Thus, after simplifying the Gaussian factors, the Schrödinger equation from Lemma \ref{lem:schrodingerextension} becomes: 

\begin{align}
    &i\left(\dot a(t)z^2 + \dot b(t)z + \dot c(t)\right)\prod_{k=1}^r(z-\lambda_k(t)) - i\sum_{k=1}^r \dot \lambda_k(t) \prod_{j\neq k} (z - \lambda_j(t))\\
    &= \left(Az^2 - B((2a(t)z + b(t))^2 + 2a(t)) - i(Cz+E)(2a(t)z + b(t)) + Dz + F - \frac{iC}{2}\right) \prod_{k = 1}^r (z-\lambda_k(t))\\
    &\quad\quad-(i(Cz + E) + B(4a(t)z + 2b(t)))\sum_{k = 1}^r \prod_{j \neq k}(z-\lambda_j(t)) - B\sum_{1 \le k \neq m \le r}\prod_{j \neq k, m}(z - \lambda_j(t)),
\end{align}
so we obtain : 
\begin{align}
    i\left(\dot a(t)z^2 + \dot b(t)z + \dot c(t)\right) - i \sum_{k = 1}^r \frac{\dot\lambda_k(t)}{z - \lambda_k(t)} =& Az^2 - B((2a(t)z + b(t))^2 + 2a(t)) - i(Cz+E)(2a(t)z + b(t)) \\
    &\quad+ Dz + F - \frac{iC}{2} \\
    &\quad-( i(Cz + E) + B(4a(t)z + 2b(t)))\sum_{k = 1}^r\frac{1}{z - \lambda_k(t)}\\
    &\quad- B\sum_{1 \le k \neq m \le r}\frac{1}{(z-\lambda_k(t))(z-\lambda_m(t))}\\
    =&(A-4Ba^2(t) - 2iCa(t))z^2 - (iCb(t) + 4Ba(t)b(t) + 2ia(t)E - D)z \\
    &\quad- B(b(t)^2 + 6a(t)) - \frac{3iC}{2} + D - iEb(t) + F\\
    &\quad-\sum_{k = 1}^r\frac{\lambda_k(t)(4Ba(t) + iC) + iE+2Bb(t) - 2B\sum_{m \neq k}\frac{1}{\lambda_k(t) - \lambda_m(t)}}{z - \lambda_k(t)},
\end{align}
where in the last line, we grouped terms in function of their degree in the variable $z$, and used the following identity, to decompose the last term of the right hand side: 
\begin{equation}
\frac{1}{(z-\lambda_k(t))(z-\lambda_m(t))} = -\frac{1}{\lambda_k(t) - \lambda_j(t)}\frac{1}{z - \lambda_k(t)} -\frac{1}{\lambda_j(t) - \lambda_k(t)}\frac{1}{z-\lambda_j(t)}.
\end{equation}
This identity is true as long as $\lambda_k(t) \neq \lambda_j(t),$ which is true for $t \in I$.
Now, to obtain the dynamics of both the Gaussian parameters $a, b, c$ and the zeros $\lambda_k(t), 1 \le k \le r$, we simply notice that both sides of the equation are partial fraction decompositions in the variable $z$, the poles and coefficients of which vary with time. By unicity of such a decomposition, we obtain: 

\begin{equation}
\begin{cases}
\displaystyle \dot a(t)= 4iB a(t)^2 - 2C a(t) - iA, \\[1.2ex]
\displaystyle \dot b(t) = 4iB a(t) b(t) - C b(t) - 2E a(t) - iD, \\[1.2ex]
\displaystyle \dot c(t) = iB \left(b(t)^2 + 6a(t)\right) - \frac{3C}{2} - iD - E b(t) - F, \\[1.2ex]
\displaystyle \dot \lambda_k(t) = \lambda_k(t)\left(-4iB a(t) + C\right) - 2iB b(t) + E + 2iB \sum_{m \neq k} \frac{1}{\lambda_k(t) - \lambda_m(t)}.
\end{cases}
\end{equation}

Finally, we show that taking the second derivative in the last equation allows us to decouple this system. We will not explicitly evaluate at $t$ to make the notation less cluttered.

\begin{align}
\ddot{\lambda}_k =& \dot{\lambda_k}(-4iBa + C) -4iB\lambda_k\frac{da}{dt} -2iB\frac{db}{dt}-2iB\sum_{m \neq k}\frac{\dot \lambda_k - \dot \lambda_m}{(\lambda_k - \lambda_m)^2}\\
=& \left(\lambda_k(-4iBa + C) - 2iBb + E + 2iB \sum_{m \neq k} \frac{1}{\lambda_k - \lambda_m}\right)(-4iBa + C) + \lambda_k(-4iB(4iBa^2 - 2Ca - iA))\\
&-2iB(4iBab - Cb -2aE -iD) -2iB\sum_{m \neq k}\frac{\dot \lambda_k - \dot \lambda_m}{(\lambda_k - \lambda_m)^2}\\
=& (-16B^2a^2 - 8iBCa +C^2) \lambda_k - 8B^2ab - 2iBCb - 4iBEa + CE + (-4iBa+C)2iB\sum_{m \neq k} \frac{1}{\lambda_k - \lambda_m} \\ 
&+ (16B^2a^2 +8iBCa- 4AB)\lambda_k + 8B^2ab + 2iBCb +4iBEa - 2BD  -2iB\sum_{m \neq k}\frac{\dot \lambda_k - \dot \lambda_m}{(\lambda_k - \lambda_m)^2} \\
=& (C^2 - 4AB)\lambda_k + CE - 2BD + 2iB(-4iBa +C)\sum_{m \neq k} \frac{1}{\lambda_k - \lambda_m} - 2iB\sum_{m \neq k}\frac{\dot \lambda_k - \dot \lambda_m}{(\lambda_k - \lambda_m)^2}\\
=& (C^2 - 4AB)\lambda_k + CE - 2BD + 2iB(-4iBa +C)\sum_{m \neq k} \frac{1}{\lambda_k - \lambda_m} \\
&- 2iB\sum_{m \neq k}\left[\frac{1}{(\lambda_k - \lambda_m)^2} \left((\lambda_k - \lambda_m)(-4iBa + C) + 2iB\left(\sum_{j \neq k}\frac{1}{\lambda_k - \lambda_j} - \sum_{j \neq m}\frac{1}{\lambda_m - \lambda_j} \right) \right)\right]\\
&= (C^2 - 4AB)\lambda_k + CE - 2BD + 4B^2\sum_{m \neq k}\left[\frac{1}{(\lambda_k - \lambda_m)^2}\left(\sum_{j \neq k} \frac{1}{\lambda_k - \lambda_j} - \sum_{j \neq m} \frac{1}{\lambda_m - \lambda_j} \right) \right].
\end{align}
This equation is independent of $a,b,c$ and concludes the proof that the system is decoupled, but we can still simplify this last line:
\begin{equation}
\sum_{m \neq k}\left[\frac{1}{(\lambda_k - \lambda_m)^2}\left(\sum_{j \neq k} \frac{1}{\lambda_k - \lambda_j} - \sum_{j \neq m} \frac{1}{\lambda_m - \lambda_j} \right) \right] = \sum_{m\neq k} \frac{1}{(\lambda_k - \lambda_m)^2} \left[\frac{2}{\lambda_k - \lambda_m} + \sum_{j \neq m, k}\frac{\lambda_{m} - \lambda_k}{(\lambda_k - \lambda_j)(\lambda_m - \lambda_j)} \right].
\end{equation}
We now notice that the inner summation symbol sums an antisymmetric function over a symmetric domain, and hence vanishes. This yields the final expression: 

\begin{equation}
\ddot \lambda_k(t) = (C^2 - 4AB)\lambda_k(t) + CE - 2BD + 8B^2\sum_{m \neq k} \frac1{(\lambda_k(t) - \lambda_m(t))^3}.
\end{equation}


\section{Resolution of the dynamical system for the zeros (proof of Theorem \ref{th:complexcalogeromoserintegration})}
\label{app:complexcalogeromoserintegration}

This proof is essentially a reformulation of the explicit integration of Calogero--Moser systems with harmonic confinement shown in \cite[Section 5]{olshanetsky1981}, extending the necessary results for the case of complex-valued solutions, and anti-harmonic motion.

Let us first place ourselves in an open time interval $I$ containing $0$, and suppose that the zeros of the wavefunction remain distinct in this interval. Then, the zeros follow the equations of motion: 
\begin{equation}
\forall{1 \le k \le r}, \ddot{\lambda}_k = (C^2 - 4AB)\lambda_k + CE - 2BD + 8B^2 \sum_{m \neq k} \frac{1}{(\lambda_k - \lambda_m)^3} 
\end{equation}

Then, observe that replacing $\lambda_k$ by $(4B^2)^{-4}\left(\lambda_k + \frac{CE - 2BD}{\omega^2}\right)$, simplifies the equation to: 
\begin{equation}
\label{eq:CMsimplifiedmotion}
\ddot{ \lambda}_k + \omega^2  \lambda_k - \sum_{m \neq k} \frac{2}{( \lambda_k -  \lambda_m)^3} = 0.
\end{equation}

We now solve this equation. To integrate this system, we use the Olshanetsky--Perelomov projection method \cite{olshanetsky1981}, which consists in the use of a larger dimension matrix space, in which the equations of motion are simpler. Projecting the obtained solution on an appropriate $n$-dimensional subspace then yields the solution to the original system.

Here, we consider the (anti-)harmonic equation in the set of $r \times r$ complex matrices: 
\begin{equation}
\label{eq:matharmmotion}
\ddot{X} + \omega^2X  = 0,
\end{equation}
for which we know that the solutions can all be written as 
\begin{equation}
\label{eq:harmsol}
Q\cos(\omega t) + \frac{P}{\omega}\sin(\omega t),
\end{equation}
where $Q$ and $P$ are arbitrary $r \times r$ complex matrices, and $\cos$ and $\sin$ are extended to the complex plane by Euler's formulas: 
\begin{equation}
\begin{cases}
    \cos z = \frac{e^{iz} + e^{-iz}}{2}\\
    \sin z = \frac{e^{iz} - e^{-iz}}{2i}.
\end{cases}
\end{equation}

We want to find the solutions which are diagonalizable in $\mathbb{C}$ for all $t$, and diagonal at $t = 0$: 
\begin{equation}
\label{eq:ulambdaeq}
X(t) = U(t)\Lambda(t)U(t)^{-1},
\end{equation}
where for all $t \in \mathbb{R}$, $U(t)$ is invertible, $U(0) = I$, and $\Lambda(t)$ is diagonal. 

\begin{lemma}
    The system of equations formed by (\ref{eq:ulambdaeq}) and (\ref{eq:matharmmotion}) is equivalent to the following modified Lax equation: 
    \begin{equation}
    \dot{L} + i[M, L] - i\omega L = 0,
    \end{equation}
    where we made the change of variable: 
    \begin{equation}\begin{cases}
    L = \dot{\Lambda} + i[M, \Lambda] + i\omega \Lambda \\
    M = -iU^{-1}\dot{U}.
    \end{cases}\end{equation}
    Moreover, we can express $X$ in function of $L(0)$ and $\Lambda(0)$: 
    \begin{equation}
    X(t) = 2\Lambda(0)e^{-i\omega t}+ \frac{L(0)}{\omega}\sin \omega t.
    \end{equation}
    
\end{lemma}
\begin{proof}
    First, let us notice that for any time-dependent matrices $V$ and $Q$, where $V$ is invertible: 
    \begin{align}
        \label{eq:laxtrick}
        \frac{d}{dt}(VQV^{-1}) &= \dot V Q V^{-1} + V \dot{Q}V^{-1} + VQ\dot{(V^{-1})}\\
        &= \dot V Q V^{-1} + V \dot{Q}V^{-1} + VQV^{-1}\dot V V^{-1} \\
        &=  V(V^{-1} \dot V Q + \dot{Q} + QV^{-1}\dot V )V^{-1}\\
        & =V(\dot Q + i[-iV^{-1}\dot{V}, Q])V^{-1}
    \end{align}

    Note that when $V = U$, we have $-iV^{-1}\dot{V} = M$ by definition. Hence, applying this formula when differentiating (\ref{eq:ulambdaeq}), we obtain: 
    \begin{equation}
    \label{eq:firstdiffULU}
    \dot X = U (L - i\omega \Lambda) U^{-1}.
    \end{equation}
    Applying the same formula when differentiating a second time yields: 
    \begin{equation}
    \ddot X = -\omega^2X =  U(\dot L - i\omega\dot \Lambda + i[M, L - i\omega\Lambda])U^{-1}.
    \end{equation}
    We now re-inject (\ref{eq:ulambdaeq}) into the latter, to obtain : 

\begin{align}
    \dot L - i\omega \dot \Lambda + i[M, L - i\omega \Lambda] + \omega^2\Lambda &= 0\\
    \dot L + i[M, L] -i\omega(\dot \Lambda + i[M,\Lambda] + i\omega\Lambda) &= 0 \\
    \dot L +i[M, L] - i\omega L& = 0.
\end{align}

    To obtain the second part of the statement, we simply evaluate equations (\ref{eq:ulambdaeq}) and (\ref{eq:firstdiffULU}) at $t = 0$, and use the generic form of the solutions (\ref{eq:harmsol}) to obtain: 
    \begin{equation}
    \begin{cases}
        Q = \Lambda(0)\\
        P = L(0)- i\omega\Lambda(0),
    \end{cases}
    \end{equation}
    so that: 
    \begin{align}
        X(t) &= \Lambda(0)\cos(\omega t) + \frac{L(0) - i\omega \Lambda(0)}{\omega} \sin (\omega t)\\
        &= \Lambda(0)e^{-i\omega t} + \frac{L(0)}{\omega}\sin(\omega t).
    \end{align}
\end{proof}

The trick now resides in showing that the equations of motion are also equivalent to an analogous modified Lax equation, for carefully constructed $L$ and $M$ depending on the variables $\lambda_k$ and $\dot{\lambda_k}$.
In \cite[Section 3]{olshanetsky1981} the expression for the coefficients of $L$ and $M$ is given for the real-valued solution case. For $1 \le j, k \le r$:
\begin{equation}
L_{jk} = (\dot { \lambda}_j + i \omega  \lambda_j)\delta_{jk} + i (1-\delta_{jk})\frac{1}{ \lambda_j -  \lambda_k}
\end{equation}
\begin{equation}
M_{jk} = \delta_{jk}\sum_{l \neq j} \frac{2}{( \lambda_j -  \lambda_l)^2} - (1-\delta_{jk})\frac{1}{( \lambda_j -  \lambda_k)^2}.
\end{equation}

We can show that the same equation holds for complex-valued solutions, for example by explicitly calculating $\dot L + i[M, L] - i\omega L$.
We deduce that the matrix $\Lambda(t) = \textrm{diag}(\lambda_1(t), \ldots, \lambda_r(t))$ is a solution of (\ref{eq:ulambdaeq}) with $U(t)$ defined to be the solution of the Cauchy problem:
\begin{equation}
\begin{cases}
    U(0) = I\\
    \dot U(t)= iU(t)M(t),
\end{cases}
\end{equation}
so that $M(t)= -iU^{-1}(t)\dot U(t)$ for all $t$.

This concludes the proof that the zeros $\lambda_k(t)$ are the eigenvalues of $X(t)$ with this choice of $L$, under the hypothesis that the zeros are distinct in an open interval around $t$ that contains $0$. We deduce the following formula for $\psi(z, t)$, still under this hypothesis: 

\begin{equation}
    \psi(z,t) = e^{a(t)z^2 + b(t)z + c(t)}\det(X(t) - zI),
\end{equation}
where $a(t), b(t), c(t)$ satisfy the differential equations of Theorem \ref{th:complexcalogeromoserzeros}. The reader may verify that we can remove the hypothesis by noticing that if we set $\psi$ equal to the right-hand side of the latter equation, then $\psi$ satisfies the Schrödinger equation (\ref{eq:schrodpolreal}), which admits a unique solution for a given initial condition. In that case, the zeros of $\psi$ are, by definition, the eigenvalues of $X(t)$.

\section{Initial conditions leading to real-valued zeros (proof of Theorem \ref{th:wf_have_real_zeros})}
\label{app:isolated}

For the phase-shift Hamiltonian ($A=B=\frac{1}{2}, C = D = E = F = 0$), the change of variable $\lambda_k\mapsto(4B^2)^{-4}\left(\lambda_k + \frac{CE - 2BD}{\omega^2}\right)$ used in the previous section is trivial. In that case, the matrix resolution of the Calogero--Moser system from Theorem \ref{th:complexcalogeromoserintegration} yields that the complex zeros $\lambda_k(t)$ of the extended wavefunction $\psi(z, t)$ at time $t$ are the eigenvalues of the matrix: 

\begin{equation}
    \Lambda(t) = \Lambda(0)e^{-it} + L(0)\sin(t),
\end{equation}

where $\Lambda = \text{diag}(\lambda_1, \ldots, \lambda_r)$, and 

\begin{equation}
L_{jk} = (\dot{ \lambda}_j + i  \lambda_j) \delta_{jk} + i(1-\delta_{jk})\frac{1}{ \lambda_j -  \lambda_k}. 
\end{equation}

Let us introduce the interaction matrix: 
\begin{equation}
P \coloneqq i(1-\delta_{jk})\frac{1}{\lambda_j - \lambda_k},
\end{equation}
so that the matrix resolution can now be rewritten: 

\begin{equation}
    \Lambda(t) = \Lambda(0)\cos(t) + \dot \Lambda(0)\sin(t) + P(0)\sin(t).
\end{equation}

We can now inject the first-order condition from Theorem \ref{th:complexcalogeromoserzeros} for the case of the rotation, which is, in matrix form:
\begin{equation}
\dot \Lambda(0) = 2ia(0)\Lambda(0).
\end{equation}

This yields: 
\begin{equation}
\label{eq:matressplit}
\Lambda(t) = (\cos(t) + 2ia(0)\sin(t))\Lambda(0) + \sin(t)P(0).
\end{equation}

Note that the first term on the right-hand side is a diagonal matrix, the entries of which follow an ellipse centered at the origin. It follows that each of these entries crosses the real axis exactly twice between $t = 0$ and $t = 2\pi$. The second term is a fully off-diagonal matrix, that we would like to treat as a perturbation of the main ellipse trajectory. To this aim, we use the following result:

\begin{theo}[Gershgorin circle theorem \cite{gers31}]
    Let $A$ be an $n \times n$ complex matrix. For $1 \le i \le n$, define the $i$-th Gershgorin disc: 
    \begin{equation}
    D_i \coloneq \{z \in \mathbb{C} \mid \vert z - A_{ii}\vert \le \sum_{j\neq i} \vert A_{ij} \vert\}.
    \end{equation}
    Then, if the union of $k$ of the discs $D_i$ is disjoint from the union of the $n-k$ other discs, then the former union contains exactly $k$ eigenvalues of $A$, counted with their multiplicities.

    In particular, if all the discs are disjoint, then for $1 \le i \le n$, $D_i$ contains exactly one eigenvalue of $A$.
\end{theo}

For $1 \le i \le n$, and $0 \le t < 2\pi$, we denote by $D_i(t)$ the $i$-th Gershgorin disc of $\Lambda(t)$. 

Our strategy is to show that, for sufficiently far apart initial zeros $\lambda_k(0)$, the Gershgorin discs remain disjoint during the whole evolution. Hence, $\lambda_i(t) \in D_i(t)$ for all $0 \le t < 2\pi$. Moreover, if the radius of $D_i(t)$ is sufficiently small, the trajectory of $\lambda_i$ is tightly controlled by the ellipse followed by $\Lambda_{ii}$, meaning that $\lambda_i$ must cross the real axis in a small time window around the time where $\Lambda_{ii}$ does.

The radius of $D_i(t)$ is by definition: 
\begin{align}
R_i(t) &\coloneq \sum_{j \neq i}\vert \Lambda_{ij}(t) \vert \\
&= \sum_{j \neq i} \vert P_{ij}(0)\sin(t)\vert \\
&= \vert \sin(t) \vert \sum_{j \neq i}\frac{1}{\vert \lambda_i(0) - \lambda_j(0) \vert}.
\end{align}

The condition for the Gershgorin discs to be pairwise disjoint is: 

\begin{equation}
    \label{ineq:cnsgershgorindiscs}
    \forall 0\le t < 2\pi, 1 \le i < j \le r, \vert \Lambda_{ii}(t) - \Lambda_{jj}(t)\vert \ge R_i(t) + R_j(t).
\end{equation}

Let us fix $1 \le i, j \le r$. Using (\ref{eq:matressplit}), we have, for all $0 \le t < 2\pi$: 
\begin{align}
    \vert \Lambda_{ii}(t) - \Lambda_{jj}(t)\vert &= \vert (\lambda_{i}(0) - \lambda_{j}(0)) (\cos(t) + 2ia(0)\sin(t))\vert \\
    &= \vert \lambda_{i}(0) - \lambda_{j}(0) \vert \sqrt{\cos^2(t) + 4a(0)^2\sin^2(t)}.
\end{align}
Moreover, 
\begin{align}
    R_i(t) + R_j(t) &\coloneq \vert \sin(t) \vert \left(\sum_{k \neq j} \frac{1}{\vert\lambda_k(0) - \lambda_j(0)\vert} + \sum_{k \neq i}\frac{1}{\vert \lambda_k(0) - \lambda_i(0) \vert } \right).
\end{align}
Note that the desired condition (\ref{ineq:cnsgershgorindiscs}) is trivially verified when $\sin(t) = 0$. Hence, for $D_i$ and $D_j$ to never intersect, it is sufficient that for all $0 < t < 2\pi$, $t \neq \pi$: 

\begin{align}
    \vert \lambda_{i}(0) - \lambda_{j}(0) \vert \sqrt{\cos^2(t) + 4a(0)^2\sin^2(t)} &\ge \vert \sin(t) \vert \left(\sum_{k \neq j} \frac{1}{\vert\lambda_k(0) - \lambda_j(0)\vert} + \sum_{k \neq i}\frac{1}{\vert \lambda_k(0) - \lambda_i(0) \vert } \right) \\
     \vert \lambda_{i}(0) - \lambda_{j}(0) \vert \sqrt{\cot^2(t) + 4a(0)^2} &\ge \left(\sum_{k \neq j} \frac{1}{\vert\lambda_k(0) - \lambda_j(0)\vert} + \sum_{k \neq i}\frac{1}{\vert \lambda_k(0) - \lambda_i(0) \vert } \right)\\
     \frac{\vert \lambda_{i}(0) - \lambda_{j}(0) \vert}{\sum_{j \neq k}\frac{1}{\vert\lambda_k(0) - \lambda_j(0)\vert} + \sum_{k \neq i}\frac{1}{\vert \lambda_k(0) - \lambda_i(0) \vert }} &\ge \frac{1}{2a(0)}.
\end{align}
This yields the simpler (but weaker) sufficient condition for all the Gershgorin discs to stay pairwise disjoint throughout the evolution: 
\begin{equation}
\min_{i,j} \vert \lambda_i (0) -\lambda_j(0)\vert \ge \sqrt{\frac{r - 1}{a(0)}}.
\end{equation}

Finally, note that at $t = 0$ and $t = \pi$, $R_i(t) = 0$, meaning that $\Lambda_{ii}(t) = \lambda_i(t)$. But by definition of $\Lambda(t)$, $\Lambda_{ii}(0) = -\Lambda_{ii}(\pi)$, and hence $\lambda_i(0) = -\lambda_i(\pi)$. Since $\lambda_i$ is continuous (it is a solution of the differential system (\ref{eq:CMsimplifiedmotion})), its trajectory must cross the real axis at some point.
Moreover, since $\lambda_i(0) = \lambda_i(2\pi)$, $\lambda_i$ must cross the real axis an even number of times, meaning at least twice. Hence, the total number of zeros (counted with multiplicities) appearing on real rotated wavefunctions is at least $2r$.

\section{Real-valued zeros for odd stellar rank (proof of Theorem \ref{th:odd_stellar_rank_crossing})}
\label{app:oddcross}

Recall the matrix resolution of the Calogero--Moser system from Theorem \ref{th:complexcalogeromoserintegration}: the complex zeros $\lambda_k(t)$ of the extended wavefunction $\psi(z, t)$ at time $t$ are the eigenvalues of the matrix: 
\begin{equation}
    \label{eq:appG_matresCM}
    \Lambda(t) = \Lambda(0)e^{-it} + L(0)\sin(t).
\end{equation}

A consequence of this formula is the following statement.
\begin{lemma}
    Consider for all $t > 0$, the set $Z(t) \coloneq \{\lambda_i(t) \mid 1 \le i \le n\}$.
    Then, the map:
    \begin{equation}
    \sigma_t\colon \begin{array}{rcl}
Z(t) & \longrightarrow & Z(t)\\
\lambda_i(t) & \longmapsto     & -\lambda_i(t+\pi)
\end{array}
\end{equation}

    is a well-defined permutation of $Z(t)$.
\end{lemma}

\begin{proof}
    The matrix resolution (\ref{eq:appG_matresCM}) implies that:
    \begin{equation}
    \label{eq:Lambda_inv}
    -\Lambda(t+\pi) = \Lambda(t),
    \end{equation}
    and hence, for $1 \le i \le n, -\lambda_i(t+\pi)$ is in $Z(t)$, meaning that $\sigma_t$ is well defined. Equation (\ref{eq:Lambda_inv}) also implies that $\sigma_t$ is surjective, and since $Z(t)$ is finite, $\sigma_t$ must be a permutation.
\end{proof}

To conclude the proof of Theorem \ref{th:odd_stellar_rank_crossing}, consider the sets:
\begin{align}
Z^+ &= \{z \in Z(0) \mid \Im(z) > 0\}\\
Z^- &= \{z \in Z(0) \mid \Im(z) < 0\}.
\end{align}
By definition, $\vert Z^+ \vert= n^+$, and $\vert Z^- \vert = n^-$.
Note that we can suppose that all elements of $Z$ have nonzero imaginary part, as the statement is trivial otherwise.

Suppose that no quadrature wavefunction of $\ket{\psi}$ has a real zero. 
Then, for all $1 \le i \le n$, $\lambda_i$ never crosses the real axis, and since $\lambda_i$ is continuous (it is a solution of the differential system (\ref{eq:CMsimplifiedmotion})), it must remain on the same side of the real axis.  In particular, $\Im(\lambda_i(0))$ and $\Im(\lambda_i(\pi))$ have the same sign.
Consequently, $\sigma_0(Z^+) \subseteq Z^-$, and $\sigma_0(Z^-) \subseteq Z^+$. But $\sigma_0$ is a permutation of $Z(0) = Z^+ + Z^-$, which means that these set inclusions must be equalities, without what we would have $\sigma_0(Z(0)) = \sigma(Z^+) \cup \sigma(Z^-) \subsetneq Z^- \cup Z^+ = Z(0)$.

Thus, $\sigma_0(Z^+) = Z^-$. Since $\sigma_0$ is a permutation, $\vert Z^+ \vert = \vert \sigma_0(Z^+)\vert$, and hence $n^+ = n^-$, which concludes.

\end{document}